\documentclass[11pt]{article}
\usepackage{amssymb,amsmath,amsthm}
\usepackage{tikz}
\usepackage{pgflibrarysnakes}
\usetikzlibrary{snakes}
\usepackage{times}
\usepackage{amstext}
\usepackage{amsopn}
\usepackage[section]{algorithm}
\usepackage{eucal}
\usepackage{latexsym}
\usepackage[rflt]{floatflt}
\usepackage{wrapfig}
\usepackage[a4paper]{geometry}

\newtheorem{theorem}{Theorem}[section]
\newtheorem{corollary}[theorem]{Corollary}
\newtheorem{lemma}[theorem]{Lemma}
 \numberwithin{equation}{subsection}

\newcommand{\N}{\ensuremath{\mathbb{N}}}

\newcommand{\heading}[1]{
  \medskip
  \noindent \textbf{#1}
}

\begin{document}

\author{Marek Cygan\thanks{Dept. of Mathematics, Computer Science and Mechanics, University of Warsaw, Poland} \and
Marcin Pilipczuk \footnotemark[1] \and
Jakub Onufry Wojtaszczyk \thanks{Institute of Mathematics, Polish Academy of Sciences}
}

\title{Beyond $O^*(2^n)$ in domination-type problems}

\maketitle





\newcommand{\maxirsetname}{ {\sc{Largest Irredundant Set}}}
\newcommand{\minirsetname}{ {\sc{Smallest Inclusion-Maximal Irredundant Set}}}
\newcommand{\maxtime}{$O(1.9657^n)$}
\newcommand{\cdstime}{$O(1.89^n)$}
\newcommand{\mintime}{$O(1.999956^n)$}

\newcommand{\capdomname}{ {\sc{Capacitated Dominating Set}}}
\newcommand{\capsimplename}{ {\sc{Simplified Capacitated Dominating Set}}}
\newcommand{\domname}{ {\sc{Dominating Set}}}
\newcommand{\capdomtime}{$O(1.89^n)$}

\begin{abstract}
  In this paper we provide algorithms faster than $O^*(2^n)$ for several NP-complete domination-type problems.
  More precisely, we provide:
  \begin{itemize}
  \item an algorithm for \capdomname{} that solves it in \capdomtime{},
  \item a branch-and-reduce algorithm solving \maxirsetname{} in \maxtime{} time,
  \item and a simple iterative-DFS algorithm for \minirsetname{} that solves it in \mintime{} time.
  \end{itemize}
  We also provide an exponential approximation scheme for \capdomname{}. All algorithms require polynomial space.
  Despite the fact that the discussed problems are quite similar to the {\sc{Dominating Set}} problem, we are not aware of any
  published algorithms solving these problems faster than the obvious $O^*(2^n)$ solution prior to this paper.
\end{abstract}

\maketitle

\newcommand{\cN}{\ensuremath{\bar{N}}}

\section{Introduction}

\heading{Notation and definitions.}
Let $G=(V,E)$ be an undirected graph. Given $F \subset E$ we write
$V(F)$ to denote the set of all endpoints of the edges in $F$.
Given $W \subset V$ by $G[W]$ we denote the subgraph induced by $W$.

We say a vertex $v \in V$ dominates $u \in V$ if $u = v$ or $uv \in E$,
i.e. a vertex dominates itself and all its neighbours.
By $\cN(v) = \{v\} \cup \{u : uv \in E\}$ we denote the set of vertices dominated by $v$.
We extend this notation to any subset $W \subset V$ by putting $\cN(W) = \bigcup_{v \in W} \cN(v)$,
we say that a set $W$ dominates a vertex $u$ if $u \in \cN(W)$. The set $\cN(W)$ is
called the {\em{closed neighbourhood}} of $W$.

The \domname{} problem asks for the smallest set that dominates the whole $V$.
In the \capdomname{} problem each vertex $v$ is additionally equipped with a number $c(v)$, which is the number of other vertices this vertex can dominate. 
Formally, we say that a set $S \subset V$ is a capacitated dominating set if there exists
$f_S:V \setminus S \to S$ such that $f_S(v)$ is a neighbour of $v$ for each $v \in V \setminus S$ and
$|f_S^{-1}(w)| \leq c(w)$ for each $w \in S$. The function $f_S$ is called a {\em{dominating function}} for the set $S$.
The \capdomname{} problem asks for the smallest possible size of a capacitated dominating set.
Note that for a given set $S$ checking whether it is a capacitated dominating set is a polynomial--time problem
which can be solved using max--flow or maximum matching techniques.

We say a set $S \subset V$ is {\em irredundant} if for any $v \in S$
there exists a vertex $u \in V$ such that $v$ dominates $u$ and $S \setminus \{v\}$ does not dominate $u$.
We call any such vertex $u$ a {\em unique} vertex for $v$.
An irredundant set is called {\em inclusion--maximal}
if it is not a proper subset of any other irredundant set.
Note that an inclusion--maximal irredundant set does not necessarily have to dominate the whole vertex set of $G$, for examples see \cite{irsetbook}.

The irredundance numbers $ir(G)$ and $IR(G)$ are defined as the smallest and largest possible cardinality of an inclusion--maximal irredundant set.
The problems \minirsetname{} and \maxirsetname{} ask for $ir(G)$ and $IR(G)$ respectively.

\heading{Our results.} In Section \ref{sec:capdom} we provide an algorithm which solves the \capdomname{} problem in
\capdomtime{}. The algorithm constructs $O^*(\binom{n}{n/3}) = O(1.89^n)$
reductions of the input graph into a \capsimplename{} problem instance (defined in Section \ref{ssec:poly}),
each solvable in polynomial time. Section \ref{ssec:capdomapx} tackles the exponential approximation of \capdomname{}.
More precisely, we provide an approximation algorithm that for a given $c \in (0, \frac{1}{3})$, in time $O^*(\binom{n}{cn}) = O^*((1/(c^c(1-c)^{1-c}))^n)$,
computes a $(\frac{1}{4c}+c)$-approximation in the case of $c < \frac{1}{4}$ or a $(2-3c)$-approximation in the case of $\frac{1}{4} \leq c < \frac{1}{3}$.

In Section \ref{sec:big} we provide a branch-and-reduce algorithm
that solves \maxirsetname{} in \maxtime{}. In Section \ref{sec:small} we consider
simple iterative depth-first search algorithm that solves \minirsetname{} and prove that
it breaks slightly the $O^*(2^n)$ barrier: we show that it works in \mintime{} time.

All algorithms in this paper require polynomial space.

\heading{Motivation and previous work.} The field of exact exponential-time algorithms
for NP-hard problems has attracted a lot of attention in the recent years (see Woeginger's survey~\cite{woeginger}).
Many difficult problems can be solved much faster than by the obvious brute-force algorithm;
examples are {\sc Independent Set} \cite{fgk:mis}, {\sc Dominating Set} \cite{fgk:domset,rooij:domsetesa09}
, {\sc Chromatic Number} \cite{bjohus:color} or {\sc Bandwidth} \cite{naszewg,feige:exp}.
A few powerful techniques have been developed, including Measure \& Conquer \cite{fgk:domset} and
inclusion/exclusion principle applications \cite{bjohus:color,bjohus:fourier,nederlof:moebius}.
However, there is still a bunch of problems for which no faster solution than the obvious one is known.
These include {\sc Subgraph Isomorphism} and {\sc Chromatic Index} which are mentioned as
open problems in~\cite{fomin:icalp09,dagstuhl}.

Superpolynomial approximation was recently considered as a way of coping with hardness of approximation of different NP-hard problems.
Results in this field include subexponential approximation algorithm for {\sc Bandwidth} on trees \cite{naszicalp,feigetalwar:bw} and
exponential approximation schemes for {\sc Chromatic Number} \cite{francuzy, kowalik:apx} or {\sc Bandwidth} on arbitrary graphs \cite{naszicalp, gaspers:bw}.

Finding an algorithm faster than $O^*(2^n)$ for {\sc Dominating Set} was an open problem until 2004. Currently
the fastest algorithm by van Rooij et al. runs in $O(1.5048^n)$ \cite{rooij:domsetesa09}. However, many other
graph-dominating problems still lack faster solutions. The problems of solving \capdomname{},
\minirsetname{} and \maxirsetname{} faster than obvious $O^*(2^n)$ algorithms were posted by van Rooij in 2008 \cite{dagstuhl},
\capdomname{} was also posted as an open problem at {\sc IWPEC 2008}.
In this paper we present algorithms providing answers to those problems.

Note that at first glance breaking $O^*(2^n)$ barrier for the \capdomname{} problem seems a hard task, since even
the brute-force $O^*(2^n)$ algorithm involves matching or max--flow techniques. Moreover,
from the parameterized point of view, Dom et al. \cite{dom:capdomparam} showed that this problem
is $W[1]$-hard when parameterized by both treewidth and solution size, and Bodlaender et al. \cite{blp:capdomplanar} showed
that even the planar version parameterized by the solution size is also $W[1]$-hard.

Our algorithm for \capdomname{} is somewhat similar to one of the first algorithms to break $O^*(2^n)$ for the classical
{\sc Dominating Set} problem, namely the algorithm of Randerath and Schiermeyer \cite{schiermeyer:domset}. Their algorithm
also involves matching arguments and our algorithm, applied to {\sc Dominating Set}, can be viewed as a simplification
of their algorithm. However we do not know whether their algorithm could be used to solve the \capdomname{} problem.

Downey et al. \cite{downey:irset} showed that the problem of finding an irredundant set of size $k$
is $W[1]$-hard when parameterized by $k$ and FPT when parameterized by $n-k$.
However, Raman and Saurabh proved that this problem become FPT in graphs without small cycles \cite{raman-saket:irset}.
It is worth noticing that irredundant sets attract some attention from
the graph-theoretical point of view \cite{irsetbook}.

There is a large number of problems for which the fastest current known algorithm requires
$O^*(2^n)$ time, including {\sc Chromatic Number} and {\sc TSP}.
Koivisto \cite{dagstuhl} conjectured
that there might exist some negative results, i.e. reductions in the following manner:
if one problem can be solved in $O^*(c^n)$ time for some $c <2$, then another problem can be solved
in $O^*(d^n)$ time for some $d<2$.
This paper shows that the problems considered by us are somewhat easier and admit solutions faster than $O^*(2^n)$.
We would like to note that our techniques in Section \ref{sec:small} seem a bit similar to those
that lead to $O^*((2-\varepsilon)^n)$-time algorithms for {\sc Dominatic Number} \cite{bjorklund:trimmed} and {\sc TSP} \cite{bjorklund:tsp}
in graphs of bounded degree.
We think it is interesting and somewhat surprising that such techniques can be used in graphs without any degree assumption.

Recently it came to our attention that Brankovic et al.~\cite{henning:irr} were
independently working on the irredundance problems but we are not aware of the details of their results.

\heading{Acknowledgements} We would like to sincerely thank George R. R. Martin for his Song of Ice and Fire novels and O. S. Card for his Ender's series --- 
they have given us a number of enjoyable and entertaining breaks while we were working on this paper.

\section{\capdomname}\label{sec:capdom}

\subsection{\capsimplename{}}\label{ssec:poly}

In this section we introduce a simplification of the \capdomname{} problem, namely the \capsimplename{} problem, which can be solved in polynomial time.

The input of \capsimplename{} is an undirected graph $G = (V, E)$, a set $U \subset V$ and a capacity function $c:V \to \N$.
We ask for the smallest capacitated dominating set $S \subset G$ containing $U$ such that there exists a dominating function $f_S$ satisfying
\begin{equation}\label{capwar}
|f_S^{-1}(v)| \leq 1 \quad {\rm for\ each} \quad v \in V \setminus U.
\end{equation}

Let $G=(V, E)$ with $U \subset V$ and a capacity function $c:V \to \N$ be a \capsimplename{} instance.
Consider the following graph $G' = (V',E')$:
\begin{itemize}
\item for any $v \in V \setminus U$ we have $v \in V'$;
\item for any $v \in U$ we have $c(v)$ copies $v_1, v_2, \ldots, v_{c(v)}$ of $v$ in $V'$;
\item for any $v \in V \setminus U$ and $u \in U$ the edge $uv_i \in E'$ for all $i$ iff $uv \in E$;
\item for $v,w \in V \setminus U$ we have $vw \in E'$ iff $vw \in E$ and $c(v) + c(w) > 0$;
\item there are no edges of the form $v_iw_j$ or $v_iv_j$ for $v,w \in U$.
\end{itemize}

\begin{figure}[htbp]
\begin{center}
{\small{
\begin{tikzpicture}[scale=1.0]
\begin{scope}[]
\foreach \x in {-0cm}
    {
      \draw (\x, 1.8cm) -- (\x-1.25cm, 2cm) -- (\x+1.25cm, 2cm) -- (\x, 1.8cm);
      \fill (\x, 1.8cm) circle (0.1cm);
      \fill (\x-1.25cm, 2cm) circle (0.1cm);
      \fill (\x+1.25cm, 2cm) circle (0.1cm);
      \draw (0,0) -- (\x-1.25cm,2cm);
      \draw (0,0) -- (\x+1.25cm,2cm);
      \draw (0,0) -- (\x,1.8cm);1
    }
    \draw (0cm,2.35cm) node {(D,2)};
    \draw (-1.25cm,2.35cm) node {(E,0)};
    \draw (1.25cm,2.35cm) node {(F,0)};

    \draw(0,-0.35) node {(A,2)};
    \fill (0,0) circle (0.1);
    \draw (0,0) -- (-1.3cm,0);
    \fill (-1.25cm,0) circle (0.1cm);
    \draw (-1.25cm,-0.35cm) node {(H,0)};

   \fill (1.25,0) circle(0.1);
   \draw (1.25,-0.35) node{(B,3)};
   \draw(0,0) -- (2.5,0);
   \fill (2.5,0) circle (0.1);
   \draw (2.5,-0.35) node {(C,2)};
  
    \draw (2.5,0) -- (2.5cm ,1cm);
     \fill (2.5cm,1cm) circle (0.1cm);
     \draw (2.5cm,1.4cm) node {(K,0)};
     \draw (1.3cm,2cm) -- (2.5cm,1cm);
     \draw (1.3cm, 2cm) -- (2.5cm, 2cm);
     \fill (2.5cm,2cm) circle (0.1cm);
     \draw (2.5cm,2.35cm) node {(M,0)};
   
     \draw (2.5cm,0) -- (3.75cm,2cm);
     \fill (3.75cm,2cm) circle (0.1cm);

     \draw (3.75cm,2cm) -- (2.5cm,2cm);
    \draw (3.75cm,2.35cm) node {(L,1)};
    
    \foreach \x/\y in {1.25cm/0cm, 2.5cm/0cm, 0cm/0cm, 0cm/1.8cm, 3.75cm/2cm}
    {
      \draw (\x,\y) circle (0.18cm);
    }

   \draw[very thick, dashed, ->] (3.5, 1.0) -- (5.5, 1.0); 
   \draw(4.5cm,0.6cm) node {U=\{A,B,C\}};
\end{scope}
\begin{scope}[xshift=6.0cm]
    \draw[dashed] (1cm,2cm) -- (3cm,2cm);
    \draw[dashed] (5cm,2cm) -- (6cm, 2cm);
    \foreach \x in {0cm, 1cm}
    {
      \foreach \xb in {0cm, 1cm, 2cm, 3cm} 
      {
         \draw[dashed] (\x,0cm) -- (\xb,2cm);
      }
    }

 \foreach \x in {5cm, 6cm}
    {
      \foreach \xb in {4cm, 6cm} 
      {
         \draw[dashed] (\x,0cm) -- (\xb,2cm);
      }
    } 

  \foreach \xa/\xb in {0cm/0cm, 1cm/1cm, 5cm/4cm}
  {
    \draw[white,very thick] (\xa,0cm) -- (\xb, 2cm);
  }
  \draw[white, very thick] (6cm,2cm) -- (5cm, 2cm);
  \draw[white, very thick] (2cm,2cm) -- (3cm, 2cm);

  \foreach \xa/\xb in {0cm/0cm, 1cm/1cm, 5cm/4cm}
  {
    \draw[very thick] (\xa,0cm) -- (\xb, 2cm);
  }
  \draw[very thick] (6cm,2cm) -- (5cm, 2cm);
  \draw[very thick] (2cm,2cm) -- (3cm, 2cm);

	\foreach \x/\y in {0cm/H', 1cm/E', 2cm/D', 3cm/F', 4cm/K', 5cm/M', 6cm/L'}
    {
      \fill (\x, 2cm) circle (0.1cm);
      \draw(\x,2.25cm) node {\y};
    }

	\foreach \x/\y in {0cm/$A_1'$, 1cm/$A_2'$, 2cm/$B_1'$, 3cm/$B_2'$, 4cm/$B_3'$, 5cm/$C_1'$, 6cm/$C_2'$}
    {
      \fill (\x, 0cm) circle (0.1cm);
      \draw(\x,-0.35cm) node {\y};
    }

\end{scope}
\end{tikzpicture}

}}
\caption{From the simple capacitated dominating set $\{A,B,C,D,L\}$ to a matching.}
\label{fig:capreduction}
\end{center}
\end{figure}
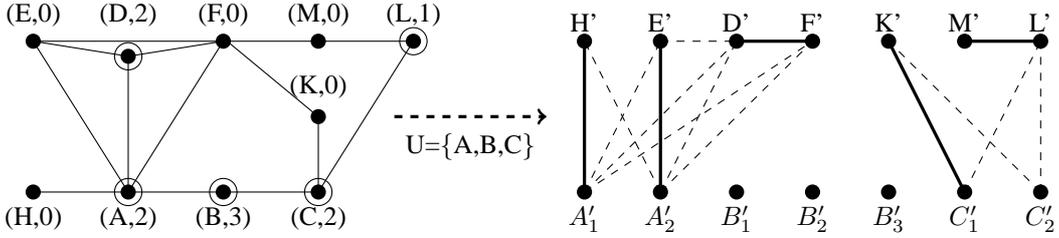

We show a correspondence between feasible solutions of \capsimplename{} in $G$ and matchings in $G'$.

\begin{lemma}\label{lem:w1}
Let $S$ be a capacitated dominating set in $G$ with dominating function $f_S$ satisfying
condition \ref{capwar}. Then one may construct in polynomial time a matching $\phi(S,f_S)$ in $G'$ satisfying
$|V| - |\phi(S, f_S)| = |S|$.
\end{lemma}

\begin{proof}
Let us define the matching $\phi(S,f_S)$ as follows:
\begin{itemize}
\item for each $v \notin S$ such that $f_S(v) \notin U$ add $vf_S(v)$ to $\phi(S, f_S)$;
\item for each $v \notin S$ such that $u = f_S(v) \in U$ add $vu_i$ to $\phi(S, f_S)$, where $u_i$ is a copy of $u$ in $G'$ and
different copies $u_i$ are chosen for different vertices $v$ with $f_S(v) = u$ (note that $|f_S^{-1}(u)| \leq c(u)$, so there are enough vertices $u_i$).
\end{itemize}
Note that every vertex $v \in V \setminus S$ is an endvertex of an edge in the matching $\phi(S, f_S)$. The second endvertex is
$f_S(v)$ (in the case $f_S(v) \notin U$) or a copy of $f_S(v)$ in $G'$ (in the case $f_S(v) \in U$). Moreover,
every edge in $\phi(S,f_S)$ has an endpoint in $V \setminus S$. Therefore $|\phi(S, f_S)| + |S| = |V|$.
\end{proof}

\begin{lemma}\label{lem:w2}
Let $M$ be a matching in $G'$. Then one may construct in polynomial time
a feasible solution $\psi(M)$ to the \capsimplename{} problem with dominating function $f_{\psi(M)}$ satisfying
 $|\psi(M)| = |V| - |M|$.
\end{lemma}

\begin{proof}
Consider the following capacitated dominating set $\psi(M)$ with dominating function $f_{\psi(M)}$:
\begin{itemize}
\item $U \subset \psi(M)$;
\item for $u \in U$ and for each $i$ such that $u_iv \in M$, we take $f_{\psi(M)}(v) = u$;
\item for any edge $vw \in M$, where $v,w \not\in U$ one of the endpoints (say $v$) has to satisfy $c(v) > 0$, we add $v$ to $\psi(M)$ and set $f_{\psi(M)}(w) = v$;
\item for any $v \not\in U$ which is not an endpoint of any edge in $M$ we add $v$ to $\psi(M)$.
\end{itemize}
It is easy to verify that the above procedure does indeed give a feasible solution to \capsimplename{}.
We have $|\psi(M)| = |V| - |M|$ since for each edge in $M$, exactly one of its endpoints does not belong to $\psi(M)$.
\end{proof}

We conclude this section with the following theorem.

\begin{theorem}
The \capsimplename{} problem can be solved in polynomial time.
\end{theorem}

\begin{proof}
By Lemmas \ref{lem:w1} and \ref{lem:w2},
to find the solution of the \capsimplename{} problem it is enough to find any maximum matching in $G'$, which can be done
in polynomial time (see e.g. \cite{muchasank:matching}).
\end{proof}

\subsection{From \capsimplename{} to \capdomname{}}

Let us start with the following simple observation.
Let $S$ be any capacitated dominating set and let $f_S$ be a dominating function for $S$.
Let
$$U_S = \{v \in S: |f_S^{-1}(v)| \geq 2\}.$$
We have
$$\sum_{v\in S} 1 + |f_S^{-1}(v)| = |S| + \sum_{v \in S} |f_S^{-1}(v)| = n,$$
thus in particular $|U_S| \leq n \slash 3$. Moreover, $S$ with the function $f_S$ is a feasible solution for
the \capsimplename{} instance with the graph $G$ and the set $U_S$. Therefore the following algorithm solves \capdomname{}:
\begin{enumerate}
\item For each $U \subset V$ satisfying $|U| \leq n\slash 3$ solve the \capsimplename{} instance
  with graph $G$ and subset $U$.
\item Return the smallest capacitated dominating set from the constructed \capsimplename{} instances.
\end{enumerate}

The \capsimplename{} problem can be solved in polynomial time and there are
$$\sum_{k=0}^{\lceil n\slash 3 \rceil} \binom{n}{k} = O^*\left( \binom{n}{\lceil n\slash 3 \rceil} \right) = O(1.89^n)$$
possible sets $U$ (i.e. sets of cardinality at most $n\slash 3$),
thus the whole algorithm works in \cdstime{} time.

\subsection{Approximating\capdomname}\label{ssec:capdomapx}

It is known that {\sc Dominating Set} is as hard to approximate as {\sc Set Cover}
and since \capdomname~problem is not easier there probably does not exists a polynomial
time algorithm solving \capdomname  with a constant approximation ratio.
If we do not have enough time to obtain an exact solution for
the \capdomname~problem we can use the following constant
approximation scheme.
Instead of investigating all subsets $U \subset V$ satisfying
$|U| \leq n\slash 3$ we can check only smaller sets,
namely $|U| \leq cn$ for some constant $c \in (0,{\frac 1 3})$.
Thus the approximation algorithm has the following form:

\begin{enumerate}
\item For each $U \subset V$ satisfying $|U| \leq cn$ solve the \capsimplename{} instance
  with graph $G$ and subset $U$.
\item Return the smallest capacitated dominating set from the constructed \capsimplename{} instances.
\end{enumerate}

\begin{theorem}
For any fixed constant $c \in (0,{\frac 1 3})$ the described algorithm runs
in $O^*(\binom{n}{cn}) = O^*((1/(c^c(1-c)^{1-c}))^n)$ time and polynomial space.
For $c \le 1 \slash 4$ the approximation ratio is at most $(\frac{1}{4c}+c)$ and
for $c \ge 1 \slash 4$ the approximation ratio is at most $2-3c$.
\end{theorem}

\begin{proof}
For each subset $U \subset V$ the algorithm uses polynomial time only, thus the time bound follows directly from the Stirling formula which can be used to bound the number of subsets $\binom{n}{cn}$.

To calculate the approximation ratio let us consider some optimal solution
$OPT \subset V$ together with a function $f_{OPT}: V \setminus OPT \rightarrow OPT$.
By $OPT_0$, $OPT_1$ and $OPT_2$ let us denote subsets of $OPT$ containing
vertices which dominate exactly zero, exactly one and at least two vertices from
$V \setminus OPT$, according to $f_{OPT}$, respectively. By $m, m_0, m_1, m_2$ we denote
the cardinalities of sets $OPT$, $OPT_0$, $OPT_1$ and $OPT_2$ respectively.

We may assume that $m_2 > cn$ since otherwise our algorithm finds the optimal solution.
By $s$ let us denote the average number of vertices from $V \setminus OPT$ which a vertex from $OPT_2$ dominates, i.e.
$s = (n-m-m_1) \slash m_2$.
Since our algorithm checks all subsets $U \subset V$ satisfying $|U| \le cn$
it obviously considers the subset $U_0 \subset OPT_2$ containing vertices which dominate the largest number of vertices from $V \setminus OPT$.
For this particular subset $U_0$ note that there exists a feasible solution to \capsimplename{} of size $m + (m_2 - cn)(s-1)$:
we take $OPT$ and for each vertex in $v \in OPT_2 \setminus U_0$ we take all but one vertices from $f_{OPT}^{-1}(v)$.

Thus the approximation ratio can be bounded by $\alpha = 1+(m_2-cn)(s-1)\slash m$.
Note that if we keep $m_2$ fixed and increase $m_0$ and $m_1$, the approximation ratio decreases --- we increase $m$ and decrease $s$ ---
therefore w.l.o.g. we may assume $m_0 = m_1 = 0$. Denoting $x = n \slash m_2$, we obtain
$\alpha \leq 1 + (1-cx)(x-2)$ for $3 \leq x \leq \frac{1}{c}$. The bound for $\alpha$ is a concave function of $x$ with maximum at
$x_0 = \frac{1}{2c} + 1$. This gives $\alpha \leq \frac{1}{4c} + c$ for $c \leq \frac{1}{4}$ and $\alpha \leq 2-3c$ for $\frac{1}{4} \leq c \leq \frac{1}{3}$.
\end{proof}

\begin{corollary}
For $c=1\slash 6$ we obtain a $5\slash 3$-approximation algorithm which runs
in $O(1.57^n)$ time and polynomial space.
\end{corollary}

This result should be compared to the trivial approximation scheme that works in $O^*(\binom{n}{cn})$ time too: iterate over all subsets of $V$ that have size at most $cn$ or at least $(1-c)n$ and return the smallest feasible solution found. However, this algorithm has an approximation factor of $\frac{1}{c}-1$, which is between $2\times$ and $4\times$ worse than our ratio.

\section{\maxirsetname}\label{sec:big}
In this section we show a branching algorithm to find the largest (in the cardinality sense) irredundant set. Notice that the largest irredundant set is necessarily inclusion--maximal.

\subsection{Bipartite independent set reduction}
Let us consider a bipartite graph $H = (W,F)$, with $W = V \cup V'$ (where $V'$ is a disjoint copy of $V$) and edges $\{u,v'\} \in F$ iff $\{u,v\} \in E$ or $u = v$. We call a set of edges $M \subset F$ in this graph {\em independent} if:
\begin{itemize}
\item no two edges in $M$ share an endpoint
\item the set $W(M)$ is an independent set in the graph $(W, F\setminus M)$, i.e. no edge connects endpoints of different edges from $M$.
\end{itemize}
Now we show a correspondence between largest irredundant sets in $G$ and largest independent sets of edges in $H$:
\begin{lemma}
If $M$ is an independent set of edges in $H$, then $S := W(M) \cap V$ is an irredundant set in $G$. Conversely, if
$S$ is an irredundant set in $G$, then there exists an independent set of edges $M \subset F$ such that $S = W(M) \cap V$.
In both cases, the independent set of edges $M$ and the irredundant set $S$ are of the same cardinality.
\end{lemma}

\begin{proof}
Let $T$ be an independent set of edges in $H$. If $\{u,v'\} \in M$, then $v'$ is not a neighbour of any $w \in M\cap V$ other than $u$, thus $v$ is a unique vertex dominated by $u$ in $G$.
Conversely, if we have an irredundant set $S$ in $G$ then letting $u(v)$ to be any unique vertex dominated by $v$ for $v \in S$ we obtain an independent set of edges $M := \{\{v,u'(v)\}: v\in S\}$ in $H$.
\end{proof}
Thus, to find the largest irredundant set in $G$ it suffices to look for the largest independent set of edges in $H$.

\begin{figure}[htbp]
\vspace{-15pt}
\begin{center}
{\small{
\begin{tikzpicture}[scale=1.0]

  \fill (0.0, 1.0) circle (0.1);
  \fill (1.0, 1.0) circle (0.1);
  \fill (2.0, 1.0) circle (0.1);
  \fill (3.0, 1.0) circle (0.1);

  \draw[thick] (0.0, 1.0) -- (3.0, 1.0);
  \draw (0.0, 1.0) circle (0.2);
  \draw (2.0, 1.0) circle (0.2);
  \draw (0.0, 1.4) node {$v_1$};
  \draw (1.0, 1.4) node {$v_2$};
  \draw (2.0, 1.4) node {$v_3$};
  \draw (3.0, 1.4) node {$v_4$};

  \draw[very thick, dashed, ->] (5.0, 1.0) -- (7.0, 1.0);

\begin{scope}[shift={(9.0, 0.0)}]
  \begin{scope}[shift={(0.0, 0.5)}]
  \fill (0.0, 1.0) circle (0.1);
  \fill (1.0, 1.0) circle (0.1);
  \fill (2.0, 1.0) circle (0.1);
  \fill (3.0, 1.0) circle (0.1);
  \draw (0.0, 1.0) circle (0.2);
  \draw (2.0, 1.0) circle (0.2);
  \draw (0.0, 1.4) node {$v_1$};
  \draw (1.0, 1.4) node {$v_2$};
  \draw (2.0, 1.4) node {$v_3$};
  \draw (3.0, 1.4) node {$v_4$};
  \end{scope}
  \begin{scope}[shift={(0.0, -0.5)}]
  \fill (0.0, 1.0) circle (0.1);
  \fill (1.0, 1.0) circle (0.1);
  \fill (2.0, 1.0) circle (0.1);
  \fill (3.0, 1.0) circle (0.1);
  \draw (0.0, 0.6) node {$v_1'$};
  \draw (1.0, 0.6) node {$v_2'$};
  \draw (2.0, 0.6) node {$v_3'$};
  \draw (3.0, 0.6) node {$v_4'$};
  \end{scope}
  \foreach \x in {0, 1, 2, 3}
  {
    \draw[thick] (\x, 1.5) -- (\x, 0.5);
  }
  \foreach \x in {0, 1, 2}
  {
    \draw[thick] (\x, 1.5) -- (\x+1.0, 0.5);
    \draw[thick] (\x+1.0, 1.5) -- (\x, 0.5);
  }
  \draw[line width=1mm] (0.0, 1.5) -- (0.0, 0.5);
  \draw[line width=1mm] (2.0, 1.5) -- (3.0, 0.5);
\end{scope}

\end{tikzpicture}
}}
\vspace{-10pt}
\caption{From the irredundant set $\{v_1, v_3\}$ to the independent set of edges $\{v_1v_1', v_3v_4'\}$.}
\label{fig:bigirsetred}
\end{center}
\end{figure}
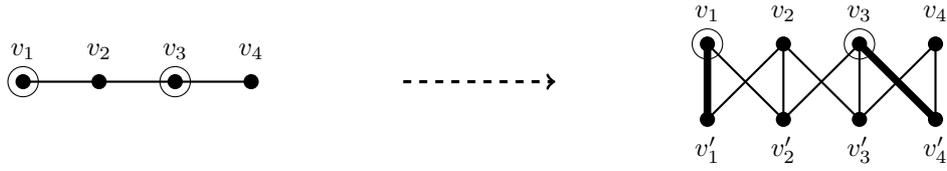

A remark to make here is that the correspondence between independent edge--sets of $H$ and irredundant sets of $G$ does not preserve inclusion--maximality --- e.g. in the graph in Figure \ref{fig:bigirsetred} there are four inclusion--maximal irredundant sets ($\{v_1,v_3\}, \{v_1,v_4\}, \{v_2,v_3\}$ and $\{v_2,v_4\}$), while in the corresponding bipartite graph the edge $v_2v_3'$ forms a singleton inclusion--maximal independent edge--set. This does no harm in the search for $IR(G)$ (as the cardinality of the set is preserved), but prevents us from using this approach to tackle the problem of $ir(G)$.

\subsection{Branch and reduce analysis}

Let $\alpha = 1.40202$. We provide an algorithm for finding the largest independent set of edges in a bipartite graph $H$ on $n$ vertices with complexity $O(\alpha^n)$. As we reduced the problem of finding the largest irredundant set in a graph with $n$ vertices to the problem of finding the largest independent edge set in a bipartite graph with $2n$ vertices, our algorithm allows us to find the largest irredundant set in a graph on $n$ vertices in $O(\alpha^{2n}) = O(1.9657^n)$ time.

The algorithm is a typical branch-and-reduce algorithm. We provide a set of reduction rules; in a single step the algorithm finds the first applicable rule and applies it.
By $T(n)$ we denote the bound on the number of steps that the algorithm requires for a graph with $n$ vertices. For each rule, we provide
a recursion inequality for $T(n)$. Since the function $T(n) = C \cdot \alpha^n$ satisfies all inequalities for any $C>0$, and every step takes polynomial time, the algorithm works in desired time.
We do simple reductions without any memoization, so the algorithm uses polynomial space. If we can spare more memory, we can have a standard trade--off between memory and computation time: we may for instance precompute answers for subgraphs of our bipartite graph having up to $n \slash 11$ vertices in $O^*(\binom{n}{n\slash 11}) = O(1.36^n)$ time and space, and then stop branching when we reach the level $n\slash 11$, this reduces the time complexity to $O(\max(1.36^n,\alpha^{10n\slash11}))=O(1.36^n)$ for finding the independent set of edges, and thus to $O(1.36^{2n})=O(1.85^n)$ for finding the largest irredundant set.

Let $H=(W,F)$ be a bipartite graph on $n$ vertices.
We use the word {\em choose} (an edge) to mean ``consider as a part of the independent edge set being built in the considered subcase'' and the word {\em drop} (a vertex) to mean ``consider that this vertex is not the endpoint of any edge in the independent edge set being built in the considered subset, so we can consider the graph without this vertex in this subcase''.
\begin{enumerate}
\item There is an isolated vertex $v \in W$. Obviously it cannot be an endpoint of an edge, so we drop it and solve the problem for $(W\setminus\{v\}, F)$.
The recursion inequality for $T(n)$ is $T(n) \geq T(n-1)$, which holds trivially.
\item There is a vertex $v$ of degree 1 in $H$, and its only neighbour is also of degree 1. These two vertices form an edge which is a part of any largest independent edge set.
Thus we choose it and solve the remaining graph on $n-2$ vertices, obtaining the trivial inequality $T(n) \geq T(n-2)$.
\item There is a vertex $v$ of degree $1$ in $H$. Let $u$ be its sole neighbour, $\deg u \geq 2$. Consider any independent edge set $M$ which contains some edge $ux$ for $x \neq v$. Then
$(M \setminus \{ux\}) \cup \{uv\}$ is also an independent set of the same cardinality. Thus we can branch out into two cases --- either we drop $u$ (and consequently drop $v$, as it becomes isolated), or we choose $uv$ and drop all the neighbours of $v$ (at least one of them). Here we obtain inequality $T(n) \geq T(n-2) + T(n-3)$.
\item There is a vertex $v$ of degree $k := \deg v \geq 8$.
We branch out, and either drop this vertex and a obtain graph with $n-1$ vertices, or choose one of the edges incident with $v$.
If we choose an edge, we remove all the neighbours of $v$ ($k$ of them) and all the neighbours of the other endpoint
(at least $2$ of them --- we already excluded vertices of degrees $0$ and $1$, and they do not coincide with the neighbours of $v$ due to bipartity) from the graph.
Thus we the obtain inequality
$$T(n) \geq T(n-1) + k\cdot T(n-k-2) \quad {\rm for} \quad k \geq 8.$$
Now we check if $T(n) = C\alpha^n$ satisfies this inequality for all $k \geq 8$. However, note that $\alpha > \frac{k+1}{k}$ for $k \geq 8$, so we only need to check the case $k=8$, which holds
by straightforward calculations.
\item There are two adjacent vertices $u$ and $v$ of degree 2 each. Let $u_1$ be the other neighbour of $u$ and $v_1$ be the other neighbour of $v$ ($u_1 \neq v_1$ due to bipartity).
We claim that there exists a largest independent edge set which contains one of the edges $uu_1$, $uv$ or $vv_1$. Indeed, consider any independent edge set $M$. Then:
\begin{itemize}
\item if neither $u_1$ nor $v_1$ is an endpoint of an edge in $M$, we may add $uv$ to $M$, preserving independence and increasing cardinality;
\item if, say, $u_1$ is the endpoint of some edge $u_1x$, but none of the three aforementioned edges belong to $M$, we may remove $u_1x$ from $M$ and add $u_1u$ instead, thus preserving cardinality and keeping independence.
\end{itemize}
Thus we consider three cases (each case consisting of choosing one of the three edges). In each case we remove the neighbours of two adjacent vertices, each of degree at least two, thus we remove at least 4 vertices in each case. We get the inequality $T(n) \geq 3T(n-4)$.
\item \label{case:deg2} There is a vertex $v$ of degree $2$. Let $u$ and $w$ be its two neighbours. Let $d_u$ and $d_w$ be degrees of $u$ and $w$ respectively and let $k$ be the number
of common neighbours of $u$ and $w$, excluding $v$. We have $3 \leq d_u, d_w < 8$ and $0 \leq k < d_u, d_w$.
We branch into four cases. In the first, the edge $uv$ is chosen, we drop the neighbours of $u$ and $v$ and obtain graph on $n-d_u-2$ vertices. In the second, we choose the edge $vw$, and similarly obtain $n-d_w-2$ vertices. In the third, we drop the vertices $u, v$ and $w$ and solve the problem for the graph on $n-3$ vertices.
In the fourth we drop $v$ and choose one of the edges coincident with $w$ and one of the edges coincident with $u$.

The case we have not considered is when there is an independent edge set $M$, where exactly one of the points $w, u$, say $w$, is coincident with an edge $wx$ from $M$. In this case, however, we may remove $wx$ from $M$, adding $wv$ instead, preserving both cardinality and independence. Thus this case need not be considered.

Now let us analyze the fourth branching option. If we choose edges $uu'$ and $ww'$ to the independent set of edges, neither $u'$ nor $w'$ is a common neighbour of $u$ and $w$.
Therefore we can choose $u'$ and $w'$ in $(d_u-k-1)(d_w-k-1)$ ways.
Let us check how many vertices are removed in each choice. From one side of the bipartite graph we remove at least $v$ and the $d_u-1+d_w-1-k$ vertices adjacent to either $u$ or $w$.
From the other we remove $u$, $w$ and at least one neighbour of $u'$ ($w'$ and $u'$ may both have degree $2$ and common neighbours, thus we cannot assume removing additionally a neighbour of $w'$).
Thus we remove at least $d_u+d_w+2-k$ vertices in total and obtain the inequality
\begin{align*}
T(n) &\geq T(n-d_u-2) + T(n-d_w-2) + T(n-3) \\
     &\quad + (d_u-k-1)(d_w-k-1)T(n-d_u-d_w-2+k).
\end{align*}
A direct check of all the possible values of $d_u$, $d_w$ and $k$ proves the inequality for $T(n) = C\alpha^n$.
\item \label{case:3i3} There are two adjacent vertices $u$ and $v$, both of degree $3$. This case is the most complicated and is considered separately in Section \ref{ssec:osobny}.
\item What remains is a case where there is a vertex $v$ of degree $3 \leq i < 8$ and all its neighbours have degree at least $4$. We branch, either
dropping $v$ or choosing one of the coincident edges. If we choose edge $vv'$, we remove at least $i+4$ vertices, thus obtaining the inequality $T(n) \geq T(n-1) + iT(n-i-4)$.
By a direct check the inequality holds for $T(n) = C\alpha^n$ for all valid values of $i$.
\end{enumerate}

\begin{figure}[htbp]
\vspace{-10pt}
\begin{center}
{\small{
\begin{tikzpicture}[scale=1.0]

\begin{scope}[shift={(0.0, 0.0)}]
  \fill (0.0, 0.0) circle (0.1);
  \fill (-1.0, 1.0) circle (0.1);
  \fill (1.0, 1.0) circle (0.1);
  \draw (-1.0, 1.0) -- (0.0, 0.0) -- (1.0, 1.0);
  \foreach \x in {-2.25, -1.75, -1.25, -0.75, -0.25, 0.25}
  {
    \draw (-1.0, 1.0) -- (\x, 2.0);
  }
  \foreach \x in {-2.25, -1.75, -1.25, -0.75, -0.25, 0.25}
  {
    \draw (1.0, 1.0) -- (-\x, 2.0);
  }
  \draw[snake=brace] (-2.4, 2.1) -- (-0.4, 2.1);
  \draw[snake=brace] (-0.4, 2.1) -- (0.4, 2.1);
  \draw[snake=brace] (0.4, 2.1) -- (2.4, 2.1);

  \draw (-1.4, 2.4) node {$(d_u-k-1)$};
  \draw (1.4, 2.4) node {$(d_w-k-1)$};
  \draw (0.0, 2.4) node {$k$};

  \draw (0.0, -0.3) node {$v$};
  \draw (-1.15, 0.75) node {$u$};
  \draw (1.15, 0.75) node {$w$};
\end{scope}

\begin{scope}[shift={(5.5, 0.0)}]
  \fill (0.0, 0.0) circle (0.1);
  \fill (-1.0, 1.0) circle (0.1);
  \fill (1.0, 1.0) circle (0.1);
  \draw (-1.0, 1.0) -- (0.0, 0.0) -- (1.0, 1.0);
  \foreach \x in {-2.25, -1.75, -1.25, -0.75, -0.25, 0.25}
  {
    \draw (-1.0, 1.0) -- (\x, 2.0);
  }
  \foreach \x in {-2.25, -1.75, -1.25, -0.75, -0.25, 0.25}
  {
    \draw (1.0, 1.0) -- (-\x, 2.0);
  }
  \draw[snake=brace] (-2.4, 2.1) -- (-0.4, 2.1);
  \draw[snake=brace] (-0.4, 2.1) -- (0.4, 2.1);
  \draw[snake=brace] (0.4, 2.1) -- (2.4, 2.1);
  \draw (-1.4, 2.4) node {$(i_1-k_u-1)$};
  \draw (1.4, 2.4) node {$(i_2-k_u-1)$};
  \draw (0.0, 2.4) node {$k_u$};

  \draw (0.0, -0.3) node {$u$};
  \draw (-1.15, 0.75) node {$u_1$};
  \draw (1.15, 0.75) node {$u_2$};
\end{scope}
\begin{scope}[shift={(10.5, -1.05)}]
  \fill (0.0, 0.0) circle (0.1);
  \fill (-1.0, 1.0) circle (0.1);
  \fill (1.0, 1.0) circle (0.1);
  \draw (-1.0, 1.0) -- (0.0, 0.0) -- (1.0, 1.0);
  \foreach \x in {-2.25, -1.75, -1.25, -0.75, -0.25, 0.25}
  {
    \draw (-1.0, 1.0) -- (\x, 2.0);
  }
  \foreach \x in {-2.25, -1.75, -1.25, -0.75, -0.25, 0.25}
  {
    \draw (1.0, 1.0) -- (-\x, 2.0);
  }
  \draw[snake=brace] (-2.4, 2.1) -- (-0.4, 2.1);
  \draw[snake=brace] (-0.4, 2.1) -- (0.4, 2.1);
  \draw[snake=brace] (0.4, 2.1) -- (2.4, 2.1);

  \draw (-1.4, 2.4) node {$(j_1-k_v-1)$};
  \draw (1.4, 2.4) node {$(j_2-k_v-1)$};
  \draw (0.0, 2.4) node {$k_v$};

  \draw (0.0, -0.3) node {$v$};
  \draw (-1.15, 0.75) node {$v_1$};
  \draw (1.15, 0.75) node {$v_2$};
\end{scope}
  \draw (5.5, 0.0) -- (10.5, -1.05);
\end{tikzpicture}
}}
\vspace{-10pt}
\caption{Notation in Cases \ref{case:deg2} and \ref{case:3i3} (i.e. Lemma \ref{lem:osobne}).}
\label{fig:bigirsetcases}
\end{center}
\end{figure}
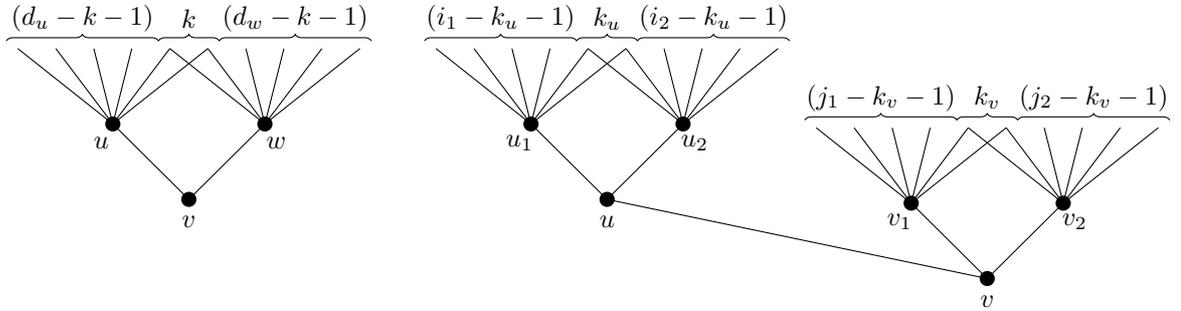

\subsection{Two adjacent vertices of degree 3}\label{ssec:osobny}
Recall the single case we have left to consider is when we have two adjacent vertices of degree $3$, and all other vertices have degrees between $3$ and $7$. Let the two adjacent vertices be $u$ and $v$, let $u_1$ and $u_2$ be the two other vertices adjacent to $u$ and let $v_1$ and $v_2$ be the two other vertices adjacent to $v$. Again, we want to take advantage of the fact that searching for the largest independent edge set we do not have to consider some of the combinations.

\begin{lemma}\label{lem:osobne}
There exists a largest independent edge set $M$ in $H$ for which either
\begin{enumerate}
\item one of the edges $uv, uu_1, uu_2, vv_1, vv_2$ appears in $M$, or
\item both $u_1$ and $u_2$ are endpoints of edges in $M$, or
\item both $v_1$ and $v_2$ are endpoints of edges in $M$ and $u_1$ and $u_2$ are not endpoints of edges in $M$.
\end{enumerate}
\end{lemma}

\begin{proof}
Take any largest independent edge set $M'$. If $M'$ satisfies the conditions above, there is nothing to prove.

If exactly one of $u_1,u_2$ is an endpoint of an edge in $M'$ (say $u_1$ is an endpoint of $u_1u_1'$), we may remove $u_1u_1'$ from $M'$ and add $uu_1$ instead. We do not lose independence (as neither $v$ nor $u_2$ was an endpoint of an edge), thus we received a new independent set of the same cardinality and satisfying our conditions. We proceed similarly if $u_1$ and $u_2$ are not endpoints of any edge in $M'$ and only one of $v_1, v_2$ is an endpoint of an edge in $M'$.

If none of $u_1, u_2, v_1, v_2$ are endpoints of edges in $M'$ we may add $uv$ to $M'$, thus contradicting the assumption $M'$ was a largest independent set.
\end{proof}

Thus we branch using the options from Lemma \ref{lem:osobne}.
There are five branches in which we choose one of the five named edges. Each edge connects two vertices of degree at least three, so in each case we obtain graph with at most $n-6$ vertices.

Let $i_1$ and $i_2$ be the degrees of $u_1$ and $u_2$ respectively, and let $j_1$ and $j_2$ be the degrees of $v_1$ and $v_2$ in $G[V \setminus \{u_1,u_2\}]$. Assume that $u_1$ and $u_2$ share $k_u$
neighbours, excluding $u$, and $v_1$ and $v_2$ share $k_v$ neighbours, excluding $v, u_1$ and $u_2$.
The next branch is when $u_1$ and $u_2$ are taken. As we take none of the named edges, we drop the vertex $u$ and then take one of the remaining edges leaving $u_1$ and one edge leaving $u_2$. We have $i_1-1-k_u$ ways to choose an edge leaving $u_1$ and $i_2-1-k_u$ ways to choose one leaving $u_2$. Let us now count the number of vertices removed in such branch. We remove $u_1$, $u_2$, $u$,
$(i_1-k_u-1) + (i_2-k_u-1) + k_u$ neighbours of $u_1$ and $u_2$, excluding $u$, and at least $2$ neighbours of the other endpoints of the chosen edges. In total, we remove $i_1+i_2-k_u+3$ vertices.
Note that we do not count $v$ to be removed, as it might be one of these $2$ neighbours of the other endpoints (one of the edges taken might be, e.g., $u_1v_1$).

The final branch is when we drop $u_1$, $u_2$, $u$ and $v$ and choose one edge leaving $v_1$ and one edge leaving $v_2$. We have $(j_1-1-k_v)(j_2-1-k_v)$ ways to choose these edges.
From one side of the bipartite graph, we remove $u_1$, $u_2$, $v$ and $(j_1-1-k_v) + (j_2-1-k_v) + k_v$ neighbours of $v_1$ and $v_2$, excluding $v, u_1$ and $u_2$. From the other side, we
remove $u$ and at least $2$ neighbours of the other endpoints of the chosen edges. In total, we remove $j_1+j_2-k_v+6$ vertices.

Therefore, in this branching rule, we obtain the following inequality for $T(n)$:
\begin{align*}
T(n) &\geq 5T(n-6) + (i_1-k_u-1)(i_2-k_u-1)T(n-i_1-i_2+k_u-3) \\
     & \quad + (j_1-k_v-1)(j_2-k_v-1)T(n-j_1-j_2+k_v-6).
\end{align*}
A direct check for all possible values of $i_1$, $i_2$, $j_1$, $j_2$, $k_u$ and $k_v$ finishes the proof that $T(n) = O(\alpha^n)$.

\section{\minirsetname}\label{sec:small}

In this section we focus on solving \minirsetname{}. We show a simple
iterative-DFS algorithm that requires polynomial space and prove that
it works in \mintime{} time. W.l.o.g. we may assume that $G$ contains
no isolated vertices, since they need to be included in any inclusion-maximal
irredundant set.

Let $\mathcal{F}_k$ be a family of irredundant sets in $G$ of size not greater than
$k$. Note that checking if a set is an (inclusion-maximal) irredundant set can
be done in polynomial time. Moreover, the family of irredundant sets is closed
under taking subsets. Therefore
$\mathcal{F}_k$ can be enumerated in $O^*(|\mathcal{F}_k|)$ time
and polynomial space by simple depth-first search algorithm.
The simple iterative-DFS algorithm
enumerates $\mathcal{F}_k$ for $k=0,1,2,\ldots,n$ until it finds
a inclusion-maximal irredundant set. Now we prove that it works
in \mintime{} time.

First, as a warm-up, let us make a $O^*((2-\varepsilon_\Delta)^n)$ bound for graphs with maximum degree bounded by $\Delta$,
where $\varepsilon_\Delta$ depends on $\Delta$. Construct a set $A \subset V$ greedily: repeatedly add any vertex $v \in V$ to $A$ and remove
from $V$ all vertices distant by at most $2$ from $v$. At each step, at most $1 + \Delta + \Delta(\Delta-1) = 1 + \Delta^2$ vertices are removed, therefore
$|A| \geq n / (1 + \Delta^2)$. The set $A$ is an independent set; moreover, closed neighbourhoods $\{\cN(v): v \in A\}$ are disjoint.
Note that any independent set $S$ cannot contain the whole $\cN(v)$ for any non-isolated vertex $v$, therefore:
$$|\mathcal{F}_n| \leq 2^n \prod_{v \in A} \frac{2^{|\cN(v)|}-1}{2^{|\cN(v)|}} \leq 2^n \left( \frac{2^{\Delta+1}-1}{2^{\Delta+1}}  \right)^\frac{n}{1+\Delta^2} =
(2-\varepsilon_\Delta)^n,$$
and the time bound for the algorithm follows.

\begin{wrapfigure}[12]{R}{5.5cm}
\vspace{-20pt}
\begin{center}
{\small{
\begin{tikzpicture}[scale=1.0]
  \draw[thick,rounded corners=8pt] (1.5, 0.0) -- (3.0, 0.0) -- (3.0, 4.0) -- (0.5, 4.0) -- (0.5, 0.0) -- (1.5, 0.0);
  \begin{scope}[shift={(0.5, 0.0)}]
  \draw[thick,rounded corners=8pt] (1.0, 0.25) -- (1.75, 0.25) -- (1.75, 2.75) -- (0.25, 2.75) -- (0.25, 0.25) -- (1.0, 0.25);
  \draw (1.0, 1.5) node {$D$};
  \end{scope}
  \draw (1.75, 3.35) node {$\cN(D)$};
  \draw[thick,rounded corners=8pt] (4.5, 0.0) -- (5.5, 0.0) -- (5.5, 4.0) -- (3.5, 4.0) -- (3.5, 0.0) -- (4.5, 0.0);
  \draw[thick,rounded corners=8pt] (5.0, 0.25) -- (5.25, 0.25) -- (5.25, 3.75) -- (4.5, 3.75) -- (4.5, 0.25) -- (5.0, 0.25);
  \draw (4.0, 3.35) node {$A_1$};
  \draw (4.875, 3.35) node {$A_2$};
  \fill (4.0, 0.3) circle (0.1);
  \fill (4.0, 2.8) circle (0.1);
  \fill (4.875, 0.8) circle (0.1);
  \fill (4.875, 1.3) circle (0.1);
  \fill (4.875, 1.8) circle (0.1);
  \fill (4.875, 2.3) circle (0.1);
  \draw[dashed] (4.0, 0.3) -- (4.875, 0.8);
  \draw[dashed] (4.0, 2.8) -- (4.875, 2.3);
  \fill (2.625, 0.9) circle (0.1);
  \fill (2.625, 1.55) circle (0.1);
  \fill (2.625, 2.2) circle (0.1);
  \draw[dashed] (2.625, 0.9) -- (4.0, 0.3);
  \draw[dashed] (2.625, 0.9) -- (4.875, 1.3);
  \draw[dashed] (2.625, 1.55) -- (4.875, 1.8);
  \draw[dashed] (2.625, 2.2) -- (4.875, 2.3);
  \draw[dashed] (2.625, 2.2) -- (4.0, 2.8);
\end{tikzpicture}
}}
\caption{Notation in the proof of Lemma \ref{strukt}}
\label{fig:smallirset}
\end{center}
\end{wrapfigure}
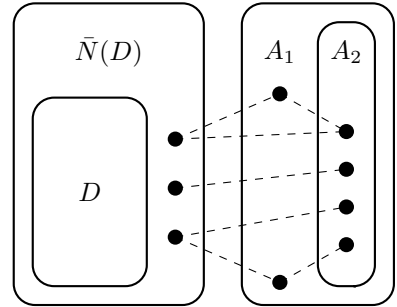

Now we show how to bypass the maximum degree assumption.
Note that if an irredundant set is a dominating set, it is inclusion--maximal.
Moreover an inclusion--minimal dominating set is an inclusion--maximal irredundant set.
Assume that $G$ admits a dominating set of size not greater than $149/300n$.
Then the algorithm stops before or at the
step $k=149/300n$ and up to this point consumes $O^*(\binom{n}{149/300n}) = O(1.999956^n)$ time.
Therefore we may consider only the case where
every dominating set in $G$ is of size greater than $149/300n$.

The following structural lemma is crucial for the analysis.

\begin{lemma}\label{strukt}
Let $G=(V, E)$ be a graph with $n$ vertices that contains
no dominating set of cardinality smaller than $149/300 n$. Then there exists
a set $A \subset V$ satisfying:
\begin{enumerate}
\item $A$ is an independent set and the neighbourhoods $\{\cN(v): v\in A\}$ are disjoint,
\item every vertex in $A$ has degree at most $6$,
\item $|A| \geq 41 n/9800$.
\end{enumerate}
\end{lemma}

\begin{proof}
We construct a dominating set $D$ greedily. Start with $D=\emptyset$.
In a single step, take any vertex $v$ that adds at least $3$ new vertices to $\cN(D)$, i.e.
$|\cN(D \cup \{v\}) \setminus \cN(D)| \geq 3$, and add $v$ to $D$.
This algorithm stops at some point and let $A_1 = V\setminus \cN(D)$,
i.e. the vertices not dominated by $D$.
For every vertex $v$ we have $|\cN(v) \cap A_1| \leq 2$, since $D$ cannot be extended any more.
In particular, every vertex in $G[A_1]$ has degree at most $1$, so
$G[A_1]$ is a graph of isolated vertices and isolated edges.
Let $A_2$ be any inclusion--maximal independent set in $G[A_1]$,
i.e. $A_2$ contains all isolated vertices of $G[A_1]$ and
one endpoint of every isolated edge. $A_2$ is an independent set in $G$, too.

Let us now note that $D \cup A_2$ is a dominating set in $G$, since $A_2$ dominates $A_1$.
Therefore $|D| + |A_2| \geq 149/300 n$. Note that, by the construction procedure
of $D$, we have $|D| \leq \frac{1}{3} |\cN(D)| = \frac{1}{3} |V \setminus A_1|$, so:
$$149/300 \leq \frac{|D| + |A_2|}{|V|} \leq \frac{1}{3} - \frac{|A_1|}{3|V|} + \frac{|A_2|}{|V|} \leq \frac{1}{3} + \frac{2}{3} \cdot \frac{|A_2|}{|V|}.$$
Therefore $|A_2| \geq 49/200 n$.

Now recall that every vertex in $V$ has at most two vertices from $A_1$ in its closed neighbourhood.
Therefore, every vertex in $V$ has at most two neighbours in $A_2$.
Let $n_7$ be the number of vertices in $A_2$ with degree at least $7$.
By counting edge endpoints we obtain that $7n_7 \leq 2(n - |A_2|) \leq 151/100 n$ and
$n_7 \leq 151/700 n$.
Let $A_3 \subset A_2$ be the set of vertices of degree at most $6$. Then $|A_3| \geq 41/1400 n$.

Now construct $A \subset A_3$ greedily. In a single step, add any $v \in A_3$ to $A$
and remove from $A_3$ the vertex $v$ and all vertices that share a neighbour with $v$ 
(recall that $A_3$ is an independent set).
Since the vertices in $A_3$ have degree at most $6$
and every vertex in $V$ is a neighbour of at most two vertices in $A_3$,
then at one step we remove at most $7$ vertices from $A_3$.
Therefore $|A| \geq 41/9800 n$.
\end{proof}
The bound for our iterative-DFS algorithm is now straightforward.
Note that for every non--isolated vertex $v$ at least one point from $\cN(v)$ does not belong to
an irredundant set. By Lemma \ref{strukt} we obtain $41/9800n$ disjoint sets $\{\cN(v): v \in A\}$,
such that all these sets are of size at most $7$ and no
$\cN(v)$ can be contained in an irredundant set.
Therefore the total number of irredundant sets is bounded by:
$$2^n \cdot \Big(\frac{2^7-1}{2^7}\Big)^\frac{41n}{9800} = O(1.99994^n).$$




\bibliographystyle{plain}
\bibliography{irredundant-stacs}

\end{document}